\documentclass[conference]{IEEEtran}

\usepackage{amsmath}
\usepackage{amssymb,mathrsfs} 
\usepackage{graphicx}

\usepackage{balance}
\usepackage{color}
\usepackage{amssymb}  
\usepackage{pifont}
\usepackage{array} 
\usepackage{tabularx} 
\usepackage{epstopdf}
\usepackage{dsfont}
\usepackage{epsfig}
\usepackage{subfig}
\usepackage{amsthm}

\makeatletter

\newcommand{\Rmnum}[1]{\expandafter\@slowromancap\romannumeral #1@}
\makeatother

\usepackage[utf8]{inputenc}
\usepackage[english]{babel}

\newtheorem{theorem}{Theorem}
\newtheorem{lemma}{Lemma}
\newtheorem{proposition}{Proposition}

\IEEEoverridecommandlockouts

\addto\captionsenglish{\renewcommand{\figurename}{Fig.}}

\usepackage[labelsep=period]{caption}

\long\def\comment#1{}

\newfont{\bbb}{msbm10 scaled 700}

\newfont{\bb}{msbm10 scaled 1100}

\newcommand{\RR}{\mbox{\bb R}}

\newcommand{\EE}{\mbox{\bb E}}

\newcommand{\yv}{{\bf y}}

\newcommand{\zerov}{{\bf 0}}

\newcommand{\Hm}{{\bf H}}
\newcommand{\Id}{{\bf I}}

\newcommand{\Sm}{{\bf S}}

\newcommand{\Vm}{{\bf V}}

\newcommand{\Tt}{\text{T}}

\newcommand{\Hc}{{\cal H}}

\newcommand{\Nc}{{\cal N}}

\newcommand{\Sc}{{\cal S}}

\newcommand{\RNum}[1]{\uppercase\expandafter{\romannumeral #1\relax}}

\newcommand{\Sigmam}{\hbox{\boldmath$\Sigma$}}

\newcommand{\diag}{{\hbox{diag}}}

\newcommand{\trace}{{\hbox{tr}}}

\def\LRT#1#2{\!
\raisebox{.2ex}{$
{{\scriptstyle\;#1}\atop{\displaystyle\gtrless}}
\atop
{\raisebox{-1.25ex}{$\scriptstyle\;#2$}}
$}
\!}

\begin{document}
%
\title{Information-Theoretic Attacks in the Smart Grid}


\author{\IEEEauthorblockN{Ke Sun$^*$, I\~naki Esnaola$^{*\S}$, Samir M. Perlaza$^{\dag\S}$, and H. Vincent Poor$^\S$}
\IEEEauthorblockA{$^*$Dept. of Automatic  Control and Systems Engineering, University of Sheffield, Sheffield S1 3JD, UK\\
$^\dag$Institut National de Recherche en Informatique et en Automatique (INRIA), Lyon, France\\
 $^\S$Dept. of Electrical Engineering, Princeton University, Princeton, NJ 08544, USA
}
\thanks{This research was supported in part by the European Commission under Marie Sk\l{}odowska-Curie Individual Fellowship No. 659316 and Euro-Mediterranean Cooperation ERA-NET project COM-MED.
The work of H. Vincent Poor was supported in part by the U.S. National Science Foundation under Grants CMMI-1435778, CNS-1702808 and ECCS-1647198.
 Ke Sun acknowledges the support of China Scholarship Council (CSC).}

}



%


\maketitle

\begin{abstract}
Gaussian random attacks that jointly minimize the amount of information obtained by the operator from the grid and the probability of attack detection are presented. The construction of the attack is posed as an optimization problem with a utility function that captures two effects: firstly, minimizing the mutual information between the measurements and the state variables; secondly, minimizing the probability of attack detection via the Kullback-Leibler divergence between the distribution of the measurements with an attack and the distribution of the measurements without an attack. Additionally, a lower bound on the utility function achieved by the attacks constructed with imperfect knowledge of the second order statistics of the state variables is obtained. 
The performance of the attack construction using the sample covariance matrix of the state variables is numerically evaluated.
The above results are tested in the IEEE 30-Bus test system.


\end{abstract}


%
\IEEEpeerreviewmaketitle

\section{Introduction}
The smart grid is expected to address the challenges posed by the move towards decentralized renewable energy generation. At the core of the smart grid lies an advanced sensing and communication network embedded in the electric power grid that provides the cyber backbone for the emerging cyberphysical energy system. The resulting system is expected to operate more efficiently while addressing some of the stability and resilience issues such as the ones causing the 2003 North American outage \cite{_final_2004}. While the increased sensing and communication enables the implementation of advanced control and management procedures, the cyber layer also opens the door to malicious attacks. The cybersecurity threats to which the smart grid is exposed are not well understood yet, and therefore, practical security solutions need to come forth as a multidisciplinary effort combining technologies such as cryptography, advanced machine learning, and information-theoretic security \cite{RS_progress_2016}.

Data injection attacks (DIAs) \cite{liu_false_2009}, which exploit the sensing infrastructure used in state estimation by the network operator, are an immediate concern. DIA strategies involve corrupting the estimate that the operator obtains of the state of the grid by tampering with the measurements produced by the sensing infrastructure. In \cite{liu_false_2009} unobservable attacks are proposed for the case in which the operator performs least squares (LS) estimation. The case in which the attacker has access to a limited number of sensors is analyzed in \cite{kim_strategic_2011, sandberg_security_2010} and \cite{sou_exact_2013}. Therein it is shown that the attacker can circumvent the access constraints using an $\ell_{1}$-minimization approach to construct a sparse attack vector. This approach is extended to the case in which multiple attackers control a subset of the nodes of the grid and coordinate to distributely construct unobservable attacks \cite{tajer_distributed_2011} and \cite{ ozay_sparse_2013}.

Bayesian frameworks are considered for transmission grids in \cite{kosut_malicious_2011} and \cite{esnaola_maximum_2016} and for distribution grids in \cite{genes_recovering_2016} by assuming a multivariate Gaussian distribution for the state variables. In this case the operator performs minimum mean square error (MMSE) estimation. Two attack  constructions that tradeoff the distortion introduced in the estimates with the probability of attack detection are presented in \cite{kosut_malicious_2011}. Maximum MMSE distortion attacks and their construction in a decentralized game-theoretic setting are studied in  \cite{esnaola_maximum_2016}.
Within an AC state estimation setting partial knowledge of the state variables is used to create unobservable attacks in \cite{liang_cyber_2014}.

The introduction of a probabilistic description of the state variables enables using information measures in the analysis of state estimation problems. For instance, a sensor placement strategy that accounts for the amount of information acquired by the sensing infrastructure is studied in \cite{li_information-theoretic_2013}. Information-theoretic security tools are employed to provide privacy guarantees in systems with smart meters \cite{sankar_smart_2013, tan_increasing_2013} and \cite{arrieta_smart_2017}. However, information-theoretic security in a smart grid state estimation context is still not well understood. In this paper, the attack construction is studied in terms of information measures to quantify the information loss that the attack causes to the operator and the probability of attack detection. The utility function that arises is analyzed in \cite{hou_effective_2014} in the context of {\it stealth} communications. Using this utility function an optimal Gaussian attack construction is obtained.
The impact of imperfect second order statistics is analyzed by considering a sample covariance estimate of the state variables as the available prior knowledge for the attacker.

The organization of the rest of this paper is as follows. Section \ref{System_Model} presents a Bayesian system model for state estimation with linearized system dynamics. A stealthy attack construction based on an information-theoretic performance measure is proposed in Section \ref{SEC:ITA}. The impact of imperfect second order statistics obtained via a limited training data set is discussed in Section \ref{SEC:AISOS}.
Section \ref{SEC:NS} numerically evaluates the performance of the proposed attack strategy on the IEEE 30-Bus test system.
The paper ends with conclusions in Section \ref{SEC:C}.

\section{System Model}\label{System_Model}
\subsection{Bayesian Framework for State Estimation}
Linearized system dynamics are considered for the state estimation problem. The resulting observation model is given by
\begin{align} \label{Equ:ObNoAtt}
Y^{M} = \Hm X^{N} + Z^{M},
\end{align}
where $X^{N} \in \RR^{N}$ is a vector of random variables describing the true state of the system;
$\Hm \in \RR^{M\times N}$ is the Jacobian of the linearized system dynamics which is determined by the power network topology and the admittances of the branches;
$Y^{M} \in \RR^{M }$ is a vector of random variables containing the measurements available to the attacker;
and $Z^{M} \in \RR^{M } $ is the additive white Gaussian noise (AWGN) introduced by the measurements \cite{abur_power_2004, grainger_power_1994}, i.e. the vector of random variables $Z^{M}$ follows a multivariate Gaussian distribution $\Nc(\zerov,\sigma^{2} \Id_{M})$.
The state variables are also assumed to follow a multivariate Gaussian distribution denoted by
\begin{align} \label{Equ:MG_sv}
  X^{N} \sim \Nc(\zerov,\Sigmam_{X\!X}),
\end{align}
where $\Sigmam_{X\!X}$ is the covariance matrix of the state variables.
Consequently, from (\ref{Equ:ObNoAtt}), the measurement vector also follows a multivariate Gaussian distribution given by
\begin{align}
  Y^{M} \sim \Nc(\zerov,\Sigmam_{Y\!Y}),
\end{align}
where $\Sigmam_{Y\!Y} = \Hm\Sigmam_{X\!X}\Hm^{\Tt} + \sigma^{2}\Id_{M} $ is the covariance matrix of the measurements.

Given the stochastic nature of the state variables, it is reasonable for the attacker to pursue a stochastic attack construction strategy.
The performance of the attack is therefore assessed in terms of the average performance that is achieved with multiple attack realizations.
In the following, an attack vector independent of the state variables is constructed following a multivariate Gaussian distribution denoted by
\begin{align}
  A^{M} \sim  \Nc (\zerov,\Sigmam_{A\!A}),
\end{align}
where $\Sigmam_{A\!A}$ is the covariance matrix of the attack vector.
It is worth noting that the independence of the attack vector with respect to the state variables implies that the attacker does not
need to know the joint distribution of the state variables and the measurements to construct the attack vector.
Knowledge of the second order moments of the state variables and the variance of the AWGN introduced by the observation process suffices to construct the attack. This assumption significantly reduces the difficulty of the attack construction.

The resulting observation model for the case in which the measurements are compromised is given by
\begin{align}
\label{eq:sys_a}
Y^{M}_{A} = \Hm X^{N} + Z^{M} + A^{M},
\end{align}
where $A^{M} \in \RR^{M}$ is the attack vector \cite{liu_false_2009}.
The compromised measurements, $Y^{M}_{A}$, follow a multivariate Gaussian distribution described as
\begin{align}
  Y_{A}^{M} \sim \Nc(\zerov,\Sigmam_{Y_{A}\!Y_{A}}),
\end{align}
where $\Sigmam_{Y_{A}\!Y_{A}} = \Hm\Sigmam_{X\!X}\Hm^{\Tt} + \sigma^{2}\Id_{M} + \Sigmam_{A\!A} $.

The operator utilizes the measurements obtained from the grid to detect the presence of an attack.
The attack detection problem is cast into a hypothesis testing problem with hypotheses
\vspace{-2mm}
\begin{eqnarray}
\Hc_{0}: & & Y^M \sim \Nc(\zerov,\Sigmam_{Y_{A}\!Y_{A}}),\quad\textnormal{versus} \\
\Hc_{1}: & & Y^M \sim \Nc(\zerov,\Sigmam_{Y\!Y}).
\end{eqnarray}

The Neyman-Pearson Lemma \cite{neyman_problem_1992} states that for a fixed Type \RNum{1} probability of error,
the likelihood ratio test achieves the minimum Type \RNum{2} probability of error $\beta$, when compared with any other tests with an equal or smaller Type \RNum{1} probability of  error $\alpha$.
In view of this, a likelihood ratio test is chosen as the attack detection strategy.
The likelihood ratio test between $\mathcal{H}_{0}$ and $\mathcal{H}_{1}$ takes the following form
\begin{equation}\label{LHRT}
L(\yv) =\frac{f_{Y^{M}_{A}}(\yv)}{f_{Y^{M}}(\yv)} \ \LRT{\Hc_{0}}{\Hc_{1}} \ \tau,
\end{equation}
where $\yv \in \RR^{M}$ is a realization of the vector of random variables modelling the measurements, $f_{Y_A^M}$ and  $f_{Y^M}$ denote the probability density functions of $Y_A^M$ and  $Y^M$, respectively, and $\tau$ is the decision threshold set by the operator to meet the false alarm constraint.

\subsection{Information-Theoretic Setting}
\label{sec:IT_setting}

The aim of the attacker when tampering with the measurements is twofold:
first, to minimize the information that the operator acquires about the state variables from the grid measurements;
second, to minimize the probability of the attack being detected by the operator.
Capitalizing on the Bayesian framework, an information-theoretic criterion for the attack construction is adopted.
To satisfy the first objective, the attacker minimizes the mutual information between the state variables and the compromised measurements.
Specifically, the attacker constructs the attack vector, i.e. chooses the distribution of the attack vector, in such a way that it minimizes $I(X^{N};Y_{A}^{M})$.
This is equivalent to guaranteeing that the amount of information that the operator acquires about the state variables $X$ by observing $Y$ is minimized.
%


On the other hand,
the probability of attack detection is determined by the detection threshold $\tau$ and the distribution induced by the attack on the measurements.
Larger values of $\tau$ yield lower probability of detection.
The Chernoff-Stein Lemma \cite{cover_elements_2012} states that the logarithm of the averaged minimum value of Type \RNum{2} probability of  error $\beta$
for any Type \RNum{1} probability of  error $\alpha$ smaller than one half asymptotically converges to
the inverse of the Kullback-Leibler (KL) divergence between the distributions of the two hypotheses.
Specifically, in this Bayesian framework, for any $\epsilon  \in (0,1/2)$,
\begin{align}
\lim_{n \to \infty} \frac{1}{n} \log \beta_{n}^{\epsilon} = -D(P_{Y^{M}_{A}}||P_{Y^{M}}),
\end{align}
where $\beta_{n}^{\epsilon}$ is the minimum $\beta$ for $\alpha < \epsilon$ when $n$ $M$-dimensional samples are available and $D(\cdot ||\cdot)$ is the KL divergence.
Therefore, for the attacker, minimizing the asymptotic detection probability is equivalent to minimizing $D(P_{Y^{M}_{A}}||P_{Y^{M}})$, where $P_{Y_A^M}$ and  $P_{Y^M}$ denote the probability distributions of $Y_A^M$ and  $Y^M$, respectively. The minimization of the KL divergence ensures that the effect of the attack on the induced distribution over the measurements is minimized, i.e. the attack is stealthy \cite{hou_effective_2014}.

In this information-theoretic setting, the attacker minimizes $I(X^{N};Y_{A}^{M})$
and  $D(P_{Y^{M}_{A}}||P_{Y^{M}})$ simultaneously.
The stealthy attack construction strategy is introduced in the next section.

\section{Stealthy Information-Theoretic Attacks}
\label{SEC:ITA}
\subsection{Stealthy Attacks}
Following the approach in \cite{hou_effective_2014},
the attacker constructs the utility function $I(X^{N};Y_{A}^{M})+ D(P_{Y^{M}_{A}}||P_{Y^{M}})$ for the attack.
The attacker minimizes this utility function to disrupt the estimation and bypass the detection of the operator.
Note that,
\begin{IEEEeqnarray}{c}
{I(X^{N};Y^{M}_{A})  +  D( P_{{Y}^{M}_{A}}||P_{Y^{M}})} = D( P_{X^{N}{Y}^{M}_{A}}||P_{X^{N}}P_{Y^{M}}) \IEEEeqnarraynumspace \IEEEyesnumber ,
\end{IEEEeqnarray}
where $P_{X^NY^M_{A}}$ is the joint distribution of $(X^{N}, Y_{A}^{M})$.
Note also that the state variables and the compromised measurements follow a multivariate Gaussian distribution given by
\begin{equation}\label{gxya}
(X^{N}, Y^{M}_{A})
 \sim \Nc (\zerov, \Sigmam),
\end{equation}
where the block covariance matrix has following structure:
\begin{equation}
\Sigmam =
\left[
\begin{array}{cc}
\Sigmam_{X\!X} & \Sigmam_{X\!X}\Hm^{\Tt} \\
\Hm\Sigmam_{X\!X} & \Hm\Sigmam_{X\!X}\Hm^{\Tt} + \sigma^{2}\Id_{M} + \Sigmam_{A\!A}
\end{array} \right].
\end{equation}
In view of this, the minimization of $I(X^{N};Y_{A}^{M})+ D(P_{Y^{M}_{A}}||P_{Y^{M}})$ is posed as the following optimization problem:
\vspace{-2mm}
\begin{align} \label{Stealth_Obj}
\underset{{A^{M}}}{\text{min}} \ D( P_{X^NY^M_{A}}||P_{X^N}P_{Y^M}).
\end{align}

\subsection{Optimal Attack Construction}

\begin{proposition}{ \rm \cite{cover_elements_2012}}\label{KLMG}
The KL divergence between two $M$-dimensional multivariate Gaussian distributions $P_{0} = \Nc (\zerov,\Sigmam_{0})$ and $P_{1} = \Nc (\zerov,\Sigmam_{1})$ is given by
\begin{align}\label{KL}
D( P_{0}||P_{1})  = & \frac{1}{2} \left( \log\frac{|\Sigmam_{1}|}{| \Sigmam_{0}|} - M +\text{\rm tr}(\Sigmam_{1}^{-\!1}\Sigmam_{0})\right).
\end{align}
\end{proposition}

Combining (\ref{KL}) and (\ref{Stealth_Obj}) it follows that the optimization problem in (\ref{Stealth_Obj}) is equivalent to
\begin{align}
\underset{\Sigmam_{A\!A} \in \Sc^{M}_{+}} {\text{min}} \quad & \left[ \trace (\Sigmam_{Y\!Y}^{-\!1}\Sigmam_{A\!A}) - \log|\Sigmam_{A\!A}+\sigma^{2}\Id_{M}| \right],  \label{Stealth_Mod_Obj}
\end{align}
where $\Sc^{M}_{+}$ is the set of all $M \times M$ positive semi-definite matrices.

\begin{proposition}\label{Stealth_Model}
The optimization problem given by (\ref{Stealth_Mod_Obj}) is equivalent to minimizing a convex function within a convex set.
\end{proposition}
\begin{proof}
The trace operator is a linear operator,
and $-\log|\Sigmam_{A\!A}+\sigma^{2}\Id_{M}|$ is a convex function of the positive semi-definite matrix $\Sigmam_{A\!A}$ \cite{boyd_convex_2004}.
Therefore, the objective function in (\ref{Stealth_Mod_Obj}) is a convex function of $\Sigmam_{A\!A}$.
Since $\Sc^{M}_{+}$ forms a convex set,
the result follows immediately.
\end{proof}

\begin{theorem} \label{Stealth_OPT}
The solution to the attack construction optimization problem (\ref{Stealth_Mod_Obj})
is the covariance matrix $\Sigmam_{A\!A}^{\star} =  \Hm\Sigmam_{X\!X}\Hm^{{\text{\rm T}}}$.
\end{theorem}

\begin{proof}
Taking the derivative of the objective function (\ref{Stealth_Mod_Obj}) with respect to $\Sigmam_{A\!A}$  yields \cite{seber_matrix_2008}
\begin{IEEEeqnarray}{lll}
\IEEEeqnarraymulticol{3}{l}{
 \frac{\partial  \big( \trace(\Sigmam_{Y\!Y}^{-\!1}\Sigmam_{A\!A}) - \log|\Sigmam_{A\!A}+\sigma^{2}\Id_{M}| \big)}{\ \partial\Sigmam_{A\!A}} }
 \nonumber \\
& \  = \ & 2 \Sigmam_{Y\!Y}^{-\!1} - \diag (\Sigmam_{Y\!Y}^{-\!1})  \nonumber\\
&& \qquad - 2(\Sigmam_{A\!A}+\sigma^{2}\Id_{M})^{-\!1} + \diag(\Sigmam_{A\!A}+\sigma^{2}\Id_{M})^{-\!1}. \IEEEeqnarraynumspace
\end{IEEEeqnarray}
Notice that the only critical point is $\Sigmam_{A\!A}^{\star} = \Hm\Sigmam_{X\!X}\Hm^{{\text{\rm T}}}$.
The result follows immediately from combining this result with Proposition \ref{Stealth_Model}.
\end{proof}

Interestingly, the optimal attack construction depends only on the second order moments of the state variables.
Therefore, the historical data of the state variables is central to the attack construction.
From a practical point of view, making historical data and the topology of the grid available to the attacker poses a security thread to the operator.
However, the extent to which historical data aids the attack construction remains to be determined.
In fact, due to practical and operational constraints, it is safe to assume that the attacker gets access to only partial information about the second order statistics of the state variables.
In the next section, the attack performance is assessed when finite training data is available to the attacker.

\section{Attack Construction with Imperfect Second Order Statistics}
\label{SEC:AISOS}
In the following, the case in which only a limited number of realizations of the state variables are available to the attacker for covariance estimation is considered.
Given that the attack construction depends only on the second order moments of the state variables, it suffices for the attacker to estimate the covariance matrix of the state variables using training samples.
Since there is no other information available, it is assumed that the attacker estimates the covariance matrix via a sample covariance matrix construction.

\subsection{Sample Covariance Matrix}

Given a set of training data $\{X^N_{i}\}^{K}_{i=1}$ of $K$ realizations of the state variables,
the sample covariance matrix is given by
\begin{equation}
\Sm_{X\!X} = \frac{1}{K-1} \sum_{i=1}^{K} X^N_{i} (X^N_{i})^{\Tt},
\end{equation}
where $K$ is the number of samples and $X^N_{i}\in\mathbb{R}^N$ is the $i$-th training sample of the state variables.
The sample covariance matrix $\Sm_{X\!X}$ coverges asymptotically to the covariance matrix $\Sigmam_{X\!X}$ and is a positive semi-definite matrix with probability 1 when $K \geq N$ \cite{tulino_random_2004}. 
Due to the randomness of the training samples $\{X^N_{i}\}^{K}_{i=1}$, the resulting sample covariance estimate, $\Sm_{X\!X}$, is a random matrix with distribution $P_{S_{X\!X}}$.

When the attacker needs to estimate the statistical structure of the state variables, instead of the optimal attack with covariance matrix $\Sigmam_{A\!A}^{\star} = \Hm\Sigmam_{X\!X}\Hm^{\Tt}$, the attacker constructs an attack with the sample covariance matrix. Conditioned on the training data, the resulting attack vector is
\begin{equation}
\tilde{A}^{M} \sim \Nc (\zerov, \Sigmam_{\tilde{A}\!\tilde{A}}),
\end{equation}
where $\Sigmam_{\tilde{A}\!\tilde{A}} = \Hm\Sm_{X\!X}\Hm^{\Tt}$.
With these estimated statistics,
the KL divergence in (\ref{Stealth_Obj})  
conditioned on the covariance matrix obtained from the training data becomes
\begin{IEEEeqnarray}{cc}
&D( P_{X^NY^M_{\tilde{A}}|S_{X\!X}}||Q_{X^N\!Y^M}|P_{S_{X\!X}}) \label{LS_UF},
\end{IEEEeqnarray}
where $ P_{X^NY^M_{\tilde{A}}|S_{X\!X}}$ is the conditional joint distribution of $(X^N, Y^M_{\tilde{A}})$ with $\Sigmam_{\tilde{A}\!\tilde{A}}\! =\! \Hm\Sm_{X\!X}\Hm^{\Tt}$ and $Q_{X^N\!Y^M}=P_{X^N}P_{Y^M}$.

\subsection{Lower Bound on Conditional KL Divergence}
The following lemma shows that the objective function in (\ref{Stealth_Obj})  for exact statistics is a lower bound on the KL divergence conditioned on the training data given by (\ref{LS_UF}).

%

\begin{lemma}\label{LS_BOUND}
The conditional divergence in (\ref{LS_UF}) for the attack vector construction with covariance $\Sigmam_{\tilde{A}\!\tilde{A}} = \Hm\Sm_{X\!X}\Hm^{{\text{\rm T}}}$ is lower bounded by the divergence in (\ref{Stealth_Obj}) with $\Sigmam_{A\!A}^{\star} = \Hm\Sigmam_{X\!X}\Hm^{{\text{\rm T}}}$, that is
\begin{equation}
D( P_{X^NY^M_{\tilde{A}}|S_{X\!X}}||Q_{X^N\!Y^M}|P_{S_{X\!X}}) \geq D( P_{X^NY^M_{A^*}}||P_{X^N}P_{Y^M}),
\end{equation}
where $P_{X^NY^M_{A^*}}$ is the joint distribution of $(X^{N}, Y_{A^*}^{M})$ when the optimal attack is constructed.
\end{lemma}

\begin{proof}
We have that
\begin{IEEEeqnarray}{ll}
&D( P_{X^NY^M_{\tilde{A}}|S_{X\!X}}||Q_{X^N\!Y^M}|P_{S_{X\!X}})   \IEEEnonumber \\
&\ = D( P_{X^NY^M_{\tilde{A}}|S_{X\!X}}||Q_{X^N\!Y^M|S_{X\!X}}|P_{S_{X\!X}}) \label{eq:div_ind} \\
&\ = \EE_{S_{X\!X}}[ D( P_{X^NY^M_{\tilde{A}}|S_{X\!X}=S}||Q_{X^N\!Y^M|S_{X\!X}=S})] \\
&\ = \frac{1}{2}\EE_{S_{X\!X}}[\trace(\Sigmam_{Y\!Y}^{-\!1}\Sigmam_{\tilde{A}\!\tilde{A}})] - \frac{1}{2}\EE_{S_{X\!X}}[\log |\Sigmam_{\tilde{A}\!\tilde{A}} + \sigma^{2}\Id_{M}|]   \IEEEnonumber \\
&\ \quad \ - \frac{1}{2} \log |\Sigmam_{Y\!Y}^{-\!1}|\\
&\ \geq \frac{1}{2}\trace(\Sigmam_{Y\!Y}^{-\!1}\Sigmam_{A\!A}^{\star}) - \frac{1}{2}\log |\Sigmam_{A\!A}^{\star} + \sigma^{2}\Id_{M}| - \frac{1}{2} \log |\Sigmam_{Y\!Y}^{-\!1}| \label{LS_lb} \IEEEeqnarraynumspace\\
&\ =D( P_{X^NY^M_{A^\star}}||P_{X^N}P_{Y^M}),
\end{IEEEeqnarray}
where (\ref{eq:div_ind}) follows from the independence of $X$ and $Y$ with respect to $S_{X\!X}$ and (\ref{LS_lb}) follows from Jensen's inequality and the fact that $-\log|\Vm|$ is a convex function of $\Vm\in \Sc^{M}_{+}$. 
\end{proof}

Lemma \ref{LS_BOUND} shows that the KL divergence achieved by the attack conditioned on the training data is higher than the performance of the attack construction with exact statistics.
However, the performance of the attack constructed by the sample covariance matrix converges asymptotically in $K$ to that of the attack constructed by the exact covariance matrix.
The speed of convergence is numerically evaluated in the following section.

\section{Numerical Results} \label{SEC:NS}

The IEEE 30-Bus test system is used to simulate the DC state estimation setting in which the bus voltage magnitudes are set to 1.0 per unit.
As a result, the state estimate is obtained using the bus injections and load consumption measurements.
The Jacobian matrix $\Hm$ is determined by the branch reactances of the grid and it is computed using MATPOWER \cite{zimmerman_matpower:_2011}.

The optimal attack construction in Theorem \ref{Stealth_OPT} shows that the covariance matrix of the attack is a function of the covariance matrix of the state variables.
To simplify the simulation, a specific Toeplitz matrix with exponential decay parameter $\rho$ is adopted \cite{esnaola_maximum_2016}.
The Toeplitz matrix of dimension $N \times N$ with exponential decay parameter $\rho$ is given by $\Sigmam_{X\!X}=[s_{ij}=\rho^{|i-j|}; i, j =1, 2, \ldots, N].$
In this setting, the utility function of the optimal attack is a function of the correlation strength $\rho$ and the noise variance $\sigma^{2}$.
We define the Signal-to-Noise Ratio (SNR) to be
\begin{equation}
\textnormal{SNR}=10\log_{10}\left(\frac{\trace{(\Hm\Sigmam_{X\!X}\Hm^\textnormal{T}})}{M\sigma^2}\right).
\end{equation}
As a result, the utility function is a function of the correlation strength $\rho$ and the SNR at which the grid operates.

The performance of the optimal attack as measured by of the utility function given by (\ref{Stealth_Obj}) is shown in Fig. \ref{Fix_snr_30}.
Surprisingly, the performance of the attack is non-monotonic with the correlation strength $\rho$. Note that the maximum value of the utility function, i.e. the worst performance of the attack vector, is represented by a star.
The simulations show that higher values of SNR yield worse performance for the attacker.
Moreover, the performance of the attack is insensitive to the correlation strength, $\rho$, for a wide range of correlation values and only becomes significant when the correlation strength is large.
For low and medium range values of the SNR, the performance of the attack is governed by the SNR and the correlation strength does not play a significant role. In the high SNR regime, the performance of the attack does not change significantly with the value of the correlation strength. This observation contrasts with linearly encoded Gaussian communication systems in which the impact of correlation is significant even for the cases in which the correlation strength is low  \cite{esnaola_linear_2013}.

The tradeoff between the disruption and the probability of attack detection is shown in Fig. \ref{MI_KL_sigma_rho}.
The performance of the attack is analyzed in terms of the mutual information, $I(X^{N};Y_{A}^{M})$, and the KL divergence, $D( P_{{Y}^{M}_{A}}||P_{Y^{M}})$, that the attack induces.
Interestingly, the performance of both objectives of the utility function is similar and there is no significant difference in the effect of the SNR or the correlation strength.
This suggests that the tradeoff between disruption and detection achieved by the optimal attack construction does not change significantly with different system parameters.
It is only when the value of the correlation strength is high that the performance gain obtained in terms of mutual information grows faster than the performance gain obtained from the KL divergence improvement.
From a practical point of view, this suggests that the attacker expects to inflict a similar disruption on the grid for a given probability of detection regardless of the system parameters $\rho$ and SNR.

In the following, the performance of the optimal attack construction is numerically evaluated when imperfect second order statistics are available to the attacker.
In particular, the sample covariance matrix estimate discussed in Section \ref{SEC:AISOS} is used to assess the performance of the attack when limited training data is available.
The performance of the attack using a sample covariance matrix when SNR=10 dB and SNR=20 dB is shown in Fig. \ref{LS_SNR_10} and Fig. \ref{LS_SNR_20}, respectively.

\begin{figure}[t!]
\centering
\includegraphics[scale=0.4]{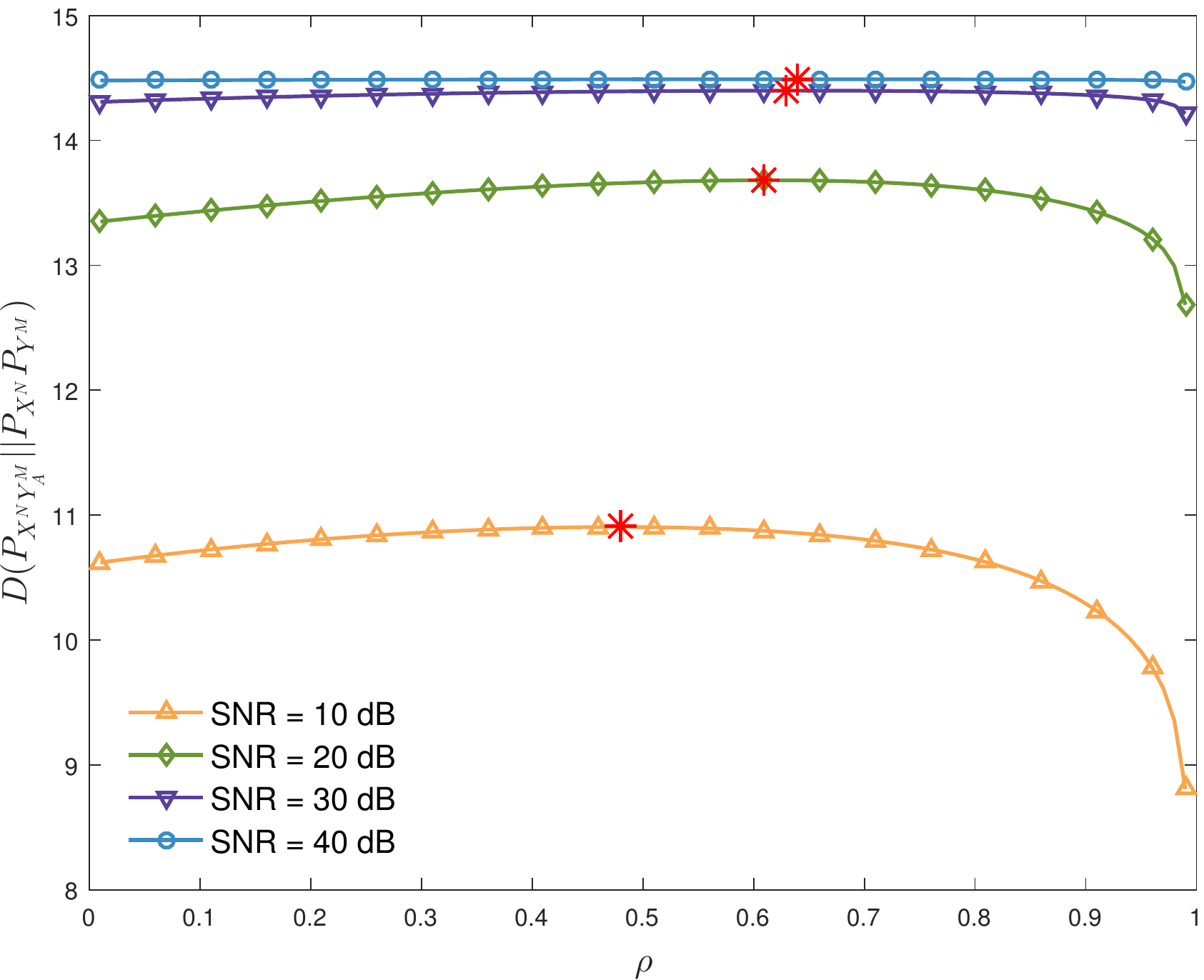}
\caption{ Performance of the optimal attack in terms of the utility function for different values of SNR in the IEEE 30-Bus test system.}
\label{Fix_snr_30}
\end{figure}

\begin{figure}[!t]
\begin{center}
\includegraphics[scale=0.4]{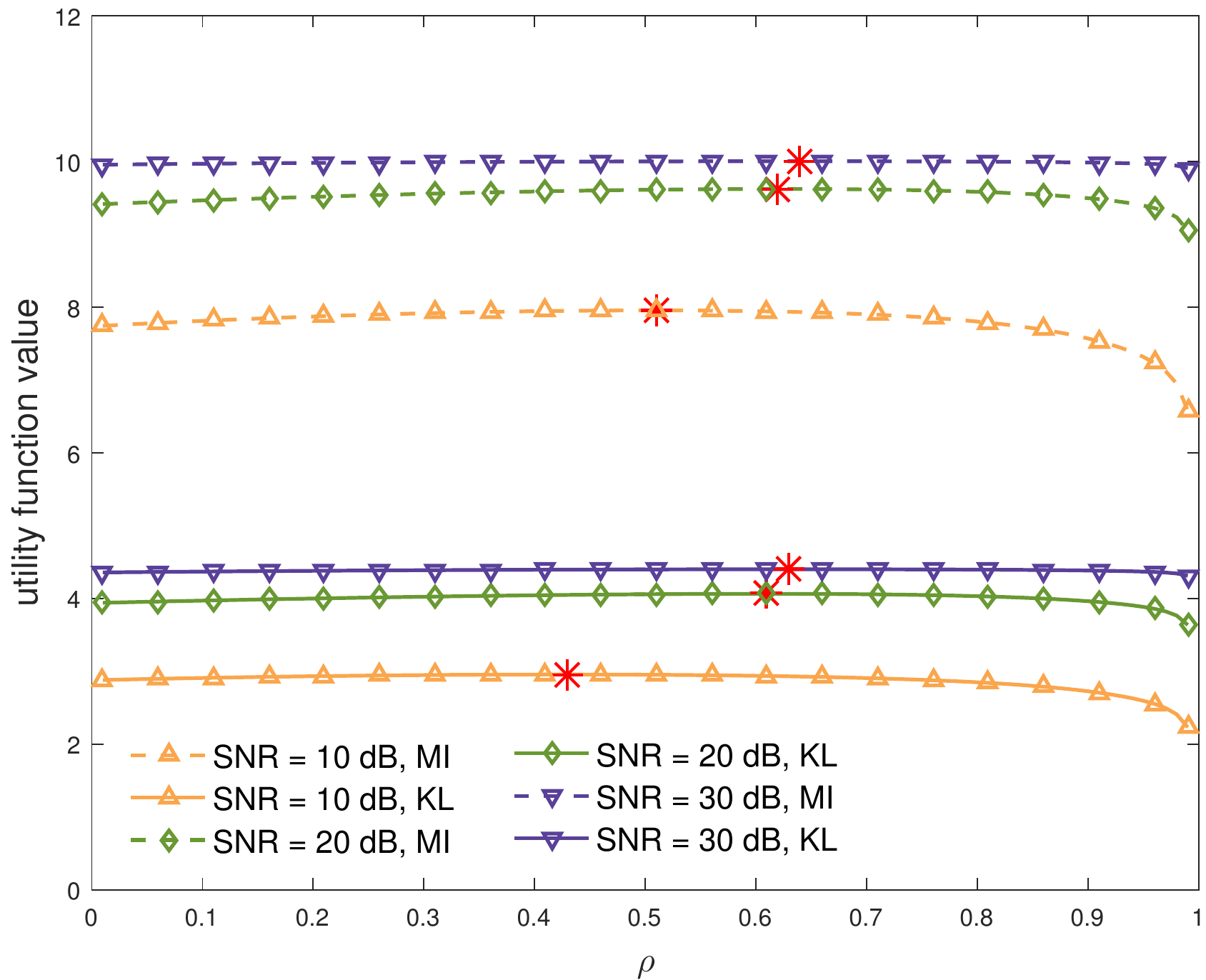}
\caption{ Performance of the optimal attack in terms of the mutual information (MI) and the KL divergence for different values of SNR in the IEEE 30-Bus test system.}
\label{MI_KL_sigma_rho}
\end{center}
\end{figure}

Therein, the dashed lines depict the performance of the optimal attack when the real covariance matrix is known, while the solid lines depict the performance of the attack constructed with the sample covariance matrix obtained with a limited number of training samples.
To guarantee the positive definiteness of $\Sm_{X\! X}$, the number of samples is larger than the size of $\Sm_{X \! X}$.
For each point, 100 realizations of the sample covariance are obtained and the utility function value is averaged over these realizations.

Interestingly, the convergence speed changes significantly for different values of the correlation strength.
The convergence is faster for lower values of the correlation strength while the impact of the SNR is not significant.
This suggests that although the performance of the optimal attack does not change significantly with respect to the correlation strength when perfect second order statistics are available, in a more realistic setting a low correlation between the state variables provides an advantageous situation for the attacker.

\begin{figure}[!t]
\begin{center}
\includegraphics[scale=0.4]{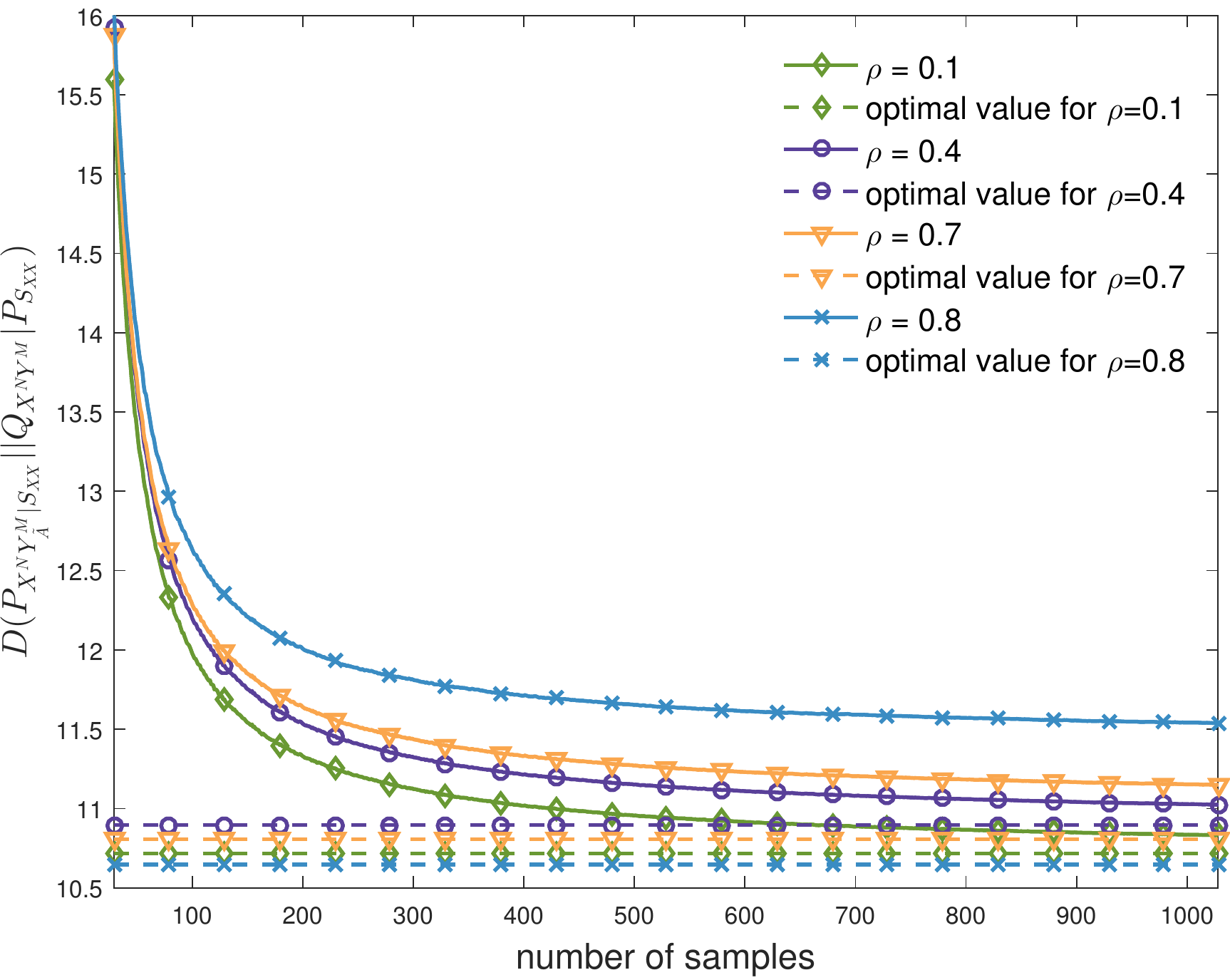}
\caption{Performance of the optimal attack for different sizes of the training set and different values of correlation strength when  SNR = 10 dB in the IEEE 30-Bus test system. }
\label{LS_SNR_10}
\end{center}
\end{figure}

\begin{figure}[!t]
\begin{center}
\includegraphics[scale=0.4]{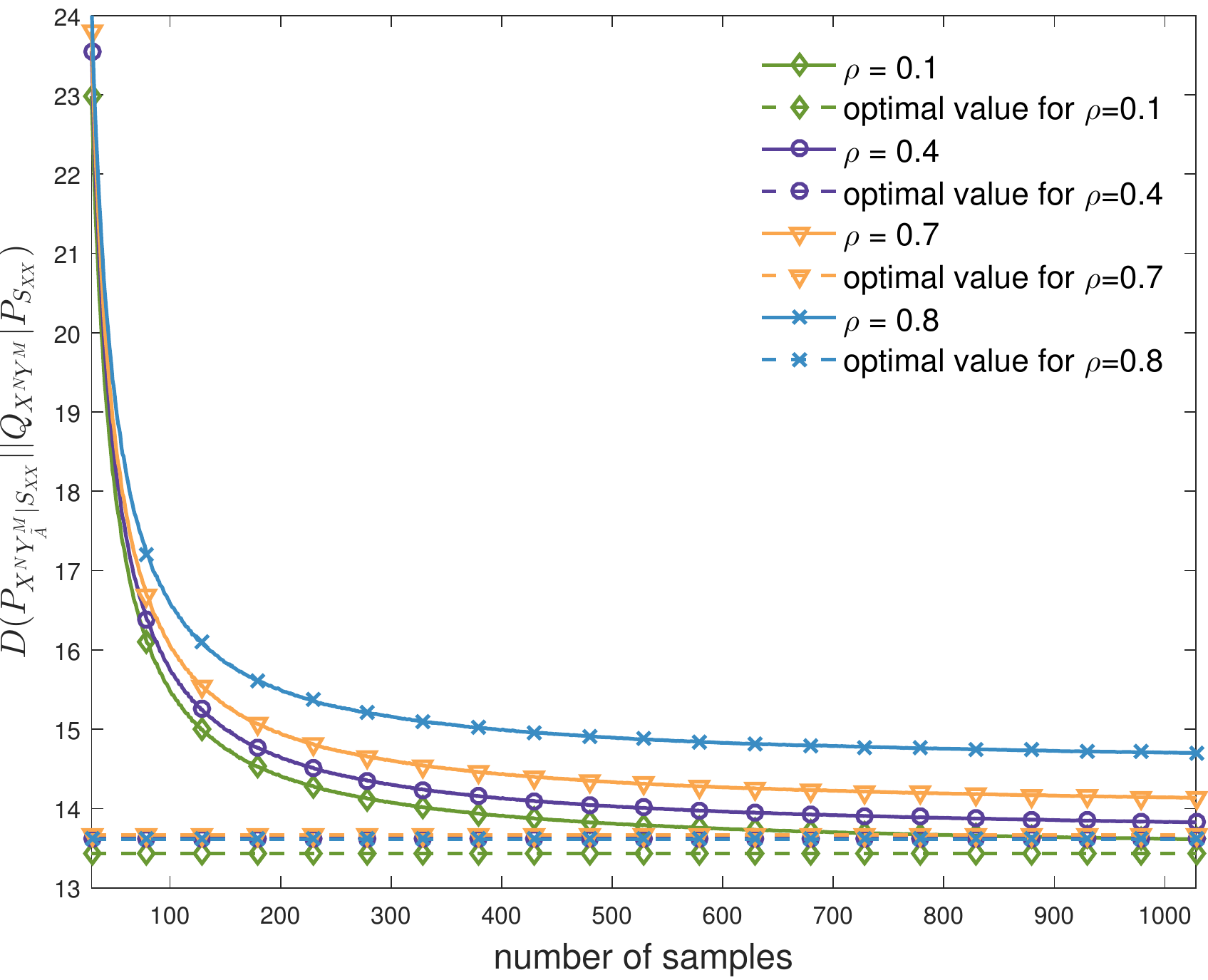}
\caption{ Performance of the optimal attack for different sizes of the training set and different values of correlation strength when  SNR = 20 dB in the IEEE 30-Bus test system }
\label{LS_SNR_20}
\end{center}
\end{figure}


Fig. \ref{LS_Attfron} shows the normalized Frobenius norm between the attack when a sample covariance matrix is used and the optimal attack with perfect second order statistics of the state variable, i.e. $\frac{||\Sigmam_{A\!A}^{\star} - \Sigmam_{\tilde{A}\!\tilde{A}}||_{F}}{||\Sigmam_{A\!A}^{\star}||_{F}}$.
\begin{figure}[!t]
\begin{center}
\includegraphics[scale=0.4]{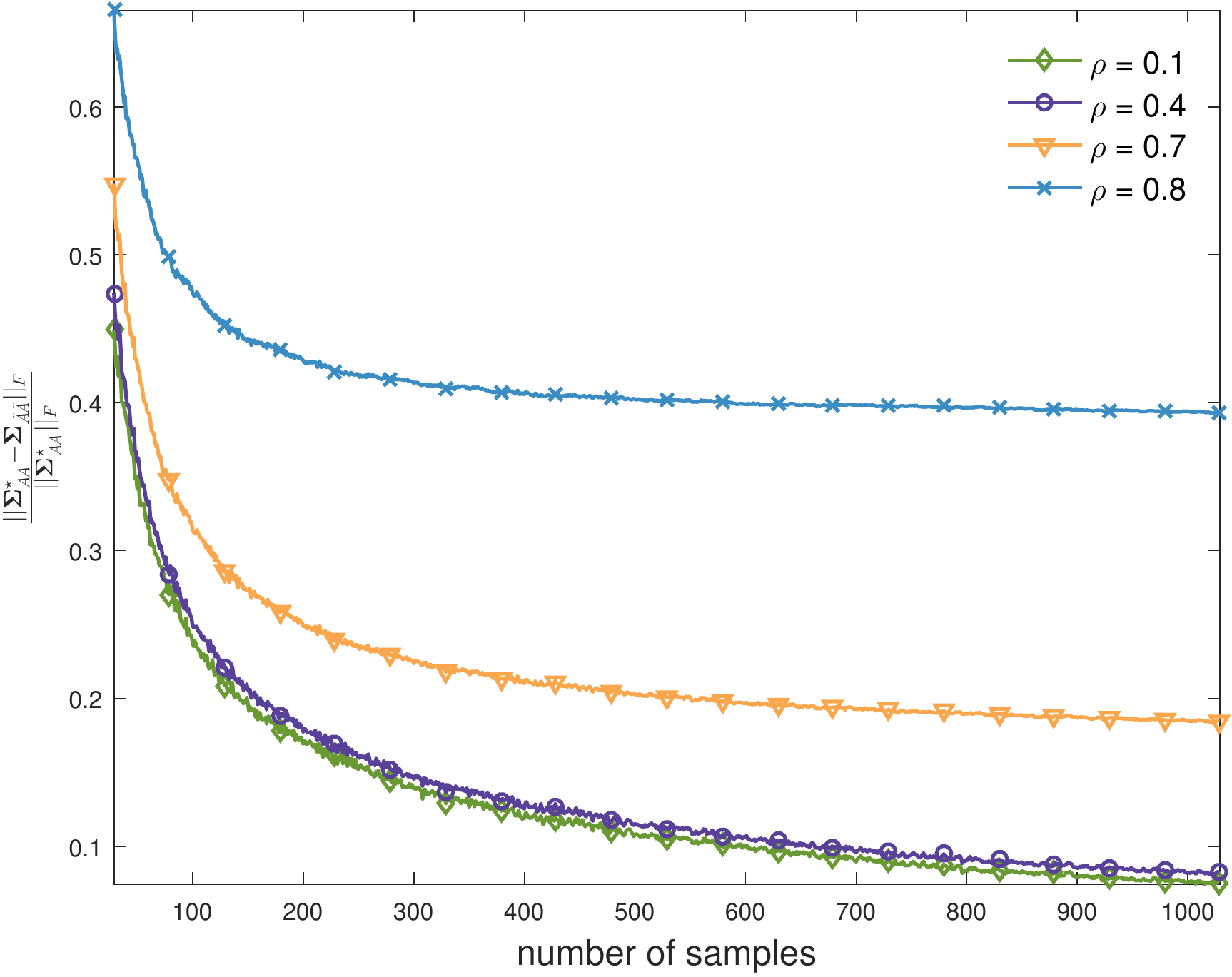}
\caption{Normalized Frobenius norm of the difference between the attack using sample covariance matrix and the optimal attack on IEEE 30-Bus system }
\label{LS_Attfron}
\end{center}
\end{figure}
The difference between the attack using the sample covariance matrix and the optimal covariance matrix decreases with the number of samples.
This implies that the attack using the sample covariance matrix converges asymptotically to the optimal covariance matrix.
However, different values of the correlation strength $\rho$ result in different convergence speeds.
Fig. \ref{Fix_snr_30} shows that when SNR = 10 dB, the performance of the optimal attack when $\rho=0.1$ and $\rho=0.8$ is almost the same.
Similarly to what is observed with the utility functions,  Fig. \ref{LS_Attfron} shows that larger values of $\rho$ converge more slowly, and as a result, the attacker needs a larger set of training samples to obtain the same performance.
Ultimately there is a trade-off between the performance of the attack and the correlation strength $\rho$ governing the state variables.
On one hand, larger correlation strength yields better attack performance.
On the other hand, larger correlation strength requires more training samples which implies the attack statistics are more difficult to learn.


\section{Conclusion} \label{SEC:C}

We have proposed a stealthy attack construction strategy within a Bayesiam framework for the smart grid.
The proposed attack construction maximizes the disruption on the state estimation that the operator obtains while minimizing the probability of attack detection.
Information-theoretic measures have been used to model the utility function for the attack construction.
Specfically, the disruption has been captured by the mutual information between the state variables and the compromised measurements, while the probability of detection has been incorporated via the KL divergence between the distributions of the measurements with and without an attack.
The resulting optimization problem has been shown to be convex and closed form expressions have been obtained.
The performance of the optimal attack construction has been numerically evaluated in an IEEE 30-Bus test system.
The impact of imperfect statistical knowledge about the state variables has also been assessed via simulations for the case in which the attacker uses a sample covariance matrix.
It has been observed that the correlation between state variables plays a critical role in the performance of the attack when limited training samples are available to the attacker.
%

\balance
\bibliographystyle{IEEEbib}
\bibliography{reference_IE}

\end{document}